\documentclass[11pt]{article}
\usepackage{amsmath,amsthm,nicefrac,amssymb}
\usepackage{amsfonts, amstext}

\usepackage[utf8]{inputenc}
\usepackage{enumitem}
\usepackage{geometry}
\usepackage[procnumbered,ruled,vlined,linesnumbered, algo2e]{algorithm2e}
\usepackage{algpseudocode}
\usepackage{amsmath,amsfonts,amsthm,amssymb,dsfont}

\usepackage{color}
\usepackage{graphicx}
\usepackage[title]{appendix}
\usepackage{thmtools,thm-restate}

\usepackage{xcolor}
\usepackage{nameref}
\definecolor{ForestGreen}{rgb}{0.1333,0.5451,0.1333}
\definecolor{DarkRed}{rgb}{0.65,0,0}
\definecolor{Red}{rgb}{1,0,0}
\usepackage[linktocpage=true,
pagebackref=true,colorlinks,
linkcolor=DarkRed,citecolor=ForestGreen,
bookmarks,bookmarksopen,bookmarksnumbered]{hyperref}
\usepackage[capitalize]{cleveref}

\newcommand{\eat}[1]{}

\declaretheorem[numberwithin=section]{theorem}
\declaretheorem[numberlike=theorem]{lemma}

\declaretheorem[name=Open Question]{oq}

\crefname{algorithm}{Algorithm}{Algorithms}
\Crefname{algorithm}{Algorithm}{Algorithms}

\theoremstyle{definition}
\declaretheorem[numberlike=theorem]{definition}
\theoremstyle{definition}

\usepackage{xparse}             %
\usepackage{xspace}              %
\usepackage[mathscr]{euscript} %

\usepackage{mleftright}         %
\usepackage{mathbbol}           %
\usepackage{fifo-stack}
\usepackage{thmtools}
\usepackage{thm-restate}
\usepackage{letltxmacro}        %
\usepackage{xpatch}              %

\usepackage[T1]{fontenc}%
\usepackage[utf8]{inputenc}%
\usepackage{xparse}
\usepackage{framed}
\usepackage{comment}

\usepackage[export]{adjustbox}
\usepackage{xparse} %
\usepackage{environ} %
\usepackage{xstring}
\usepackage{tabularx}
\usepackage{enumitem}
\usepackage[htt]{hyphenat}%
\usepackage{varwidth}     %
\usepackage{needspace}          %

\NewDocumentCommand{\cutsize}{O{\delta} g g e{_}}{%
  #1
  \IfNoValueF{#3}{_{#3}}%
  \IfNoValueF{#4}{_{#4}}%
  \IfNoValueF{#2}{\parof{#2}}
}%

\newcommand{\poly}{\operatorname{poly}} %

\newcommand{\ignore}[1]{}

\renewcommand{\paragraph}[1]{\medskip\noindent{\bf #1}\xspace}

\newcommand{\wg}{W}

\newcommand{\dist}{\operatorname{dist}}

\DeclareMathOperator{\polylog}{polylog}
\newcommand\eps{\varepsilon}
\usepackage{thm-restate}
\usepackage{placeins}
\usepackage{tikz}
\usetikzlibrary{decorations.pathreplacing}
\usepackage[font=small,width=0.85\textwidth]{caption}

\usepackage{todonotes}

\author{Aaron Bernstein\thanks{Rutgers University. Funded by NSF CAREER grant 1942010 and Google Research Grant. \texttt{bernstei@gmail.com}} \and Greg Bodwin\thanks{University of Michigan.  Funded by NSF:AF 2153680. \texttt{bodwin@umich.edu}} \and Nicole Wein\thanks{Simons Institute. This work started while the author was at DIMACS, Rutgers University. \texttt{nswein@umich.edu}}}

\title{Are there graphs whose shortest path structure requires large edge weights?}
\date{}
\begin{document}

\maketitle

\abstract{
The \emph{aspect ratio} of a (positively) weighted graph $G$ is the ratio of its maximum edge weight to its minimum edge weight.
Aspect ratio commonly arises as a complexity measure in graph algorithms, especially related to the computation of shortest paths.
Popular paradigms are to interpolate between the settings of weighted and unweighted input graphs by incurring a dependence on aspect ratio, or by simply restricting attention to input graphs of low aspect ratio.

This paper studies the effects of these paradigms, investigating whether graphs of low aspect ratio have more structured shortest paths than graphs in general.
In particular, we raise the question of whether one can generally take a graph of large aspect ratio and \emph{reweight} its edges, to obtain a graph with bounded aspect ratio while preserving the structure of its shortest paths.
Our findings are:
\begin{itemize}
\item Every weighted DAG on $n$ nodes has a shortest-paths preserving graph of aspect ratio $O(n)$.
A simple lower bound shows that this is tight.

\item The previous result does not extend to general directed or undirected graphs; in fact, the answer turns out to be \emph{exponential} in these settings.
In particular, we construct directed and undirected $n$-node graphs for which any shortest-paths preserving graph has aspect ratio $2^{\Omega(n)}$.
\end{itemize}

We also consider the \emph{approximate} version of this problem, where the goal is for shortest paths in $H$ to correspond to approximate shortest paths in $G$. We show that our exponential lower bounds extend even to this setting. We also show that in a closely related model, where approximate shortest paths in $H$ must also correspond to approximate shortest paths in $G$, even DAGs require exponential aspect ratio. 
}

\pagenumbering{gobble}




\clearpage
\pagenumbering{arabic}
\section{Introduction}

In modern graph algorithms, a popular strategy is to \emph{simplify} graphs in preprocessing before proceeding to the main part of the algorithm.
Given an input graph $G$, the high-level goal is to reduce to solving the problem on a different graph $H$ that faithfully encodes the important structural properties of $G$ (and hence solving the problem over $H$ instead of $G$ yields a good solution), but where $H$ has improved complexity measures relative to $G$ in whatever sense is important for the runtime of the algorithm to follow.
Some successful examples of this method in the literature include:

\begin{itemize}
\item \textbf{Edge Sparsifiers}, where the goal is for $H$ to have substantially fewer edges than $G$.
Well-studied examples of edge sparsifiers include spanners \cite{Pu89jacm,PU89sicomp}, preservers \cite{CE06,AB18}, flow/cut/spectral sparsifiers \cite{ST11,BK96}, etc.

\item \textbf{Vertex Sparsifiers}, where the goal is for $H$ to have substantially fewer vertices than $G$.
Well-studied examples of vertex sparsifiers include mimicking networks \cite{CDKLLPSV21, HKNR98}, terminal minor sparsifiers \cite{KNZ14}, etc. 

\item \textbf{Hopsets and Shortcut Sets}, in which $H$ has smaller \emph{hop-diameter} than $G$ while preserving its shortest path distances or reachabilities.
Hop-diameter is an important complexity measure in parallel or distributed contexts \cite{Thorup92,UY91}.

\item \textbf{Expander Decomposition}, which reduces to solving problems on graphs with a favorable \emph{expansion} parameter \cite{KVV04,ST11}.
\end{itemize}

This paper introduces a new paradigm for graph simplification, in which the goal is to minimize the complexity measure of \emph{aspect ratio}.
The aspect ratio of a graph is the multiplicative spread among its edge weights:

\begin{definition} [Aspect Ratio]
Let $G = (V, E, w)$ be a graph with positive edge weights.
Then its aspect ratio is the quantity $\frac{\max_e w(e)}{\min_e w(e)}$.
\end{definition}

Aspect ratio often enters the picture for graph algorithms where the state-of-the-art solutions for \emph{unweighted} graphs substantially outperform the solutions for \emph{weighted} graphs.
For these problems, one can sometimes extend the unweighted solution to graphs that are ``close to unweighted.''
Specifically, this could mean either (1) that the algorithm extends to the setting of weighted input graphs, but that the runtime or solution quality suffers a dependence on aspect ratio (say $\polylog r$, for an input graph of aspect ratio $r$), or (2) more simply, that the algorithm only extends to the setting of weighted input graphs of bounded aspect ratio (say $\poly(n)$, for $n$-node input graphs).
These paradigms are especially common for problems related to computation of shortest paths or distances, for several reasons: when aspect ratio is bounded one can use bucketing methods to group shortest paths by length, and algebraic techniques based on fast matrix multiplication only give speedups when the associated matrix has bounded entries.
Some concrete examples of shortest-path-related problems that have incurred a dependence on aspect ratio for these reasons
include $(1+\eps)$-hopsets \cite{BW23, KP22a}, roundtrip spanners \cite{RTZ08, ZL18, CDG20}, the All-Pairs Shortest Paths (APSP) problem \cite{Zwick02, SZ99}, dynamic shortest paths (e.g. \cite{HenzingerKN14,BernsteinGW20,ChuzhoyZ23}), distributed shorted paths (e.g. \cite{ForsterN18,CaoFR21}), and many others.
In light of this, we believe that the following two questions are natural:

\paragraph{Main Questions (Informal):}
\begin{itemize}
\item \emph{Is it generally possible to decrease the aspect ratio of a graph without changing the structure of its shortest paths?}

\item \emph{Do graphs of bounded aspect ratio have more structured shortest paths than general graphs?}
\end{itemize}

A positive resolution to either question would shed light on aspect-ratio-sensitive graph algorithms.
In fact, as we point out next, some kind of structure is guaranteed: exactly one of the two questions must be answerable in the affirmative.
To explain this point, let us formalize our model:


\begin{definition}[Shortest-Paths Preserving Graph -- See Figure \ref{fig:shppres}] \label{def:shpequiv}
    Given a graph $G = (V, E, w)$, a reweighted graph on the same vertex and edge set $H = (V, E, w_H) $ is \emph{shortest-paths preserving} if, for every shortest path $\pi$ in $H$, the sequence of nodes and edges along $\pi$ is also a shortest path in $G$.
    \footnote{If $G$ has several tied-for-shortest $s \leadsto t$ paths, one could consider two different models: that \emph{at least one} shortest $s \leadsto t$ path in $G$ must be shortest in $H$ (which is Definition \ref{def:shpequiv}), or that \emph{all} shortest $s \leadsto t$ paths in $G$ must be shortest in $H$.  All of our upper and lower bounds work for both of these definitions, so we somewhat arbitrarily use the first one.}
 
\end{definition}


\begin{figure}[h]
\centering
\begin{tikzpicture}[scale=.8]
\draw [fill=black] (0, 0) circle [radius=0.15];
\draw [fill=black] (2, -2) circle [radius=0.15];
\draw [fill=black] (2, 2) circle [radius=0.15];
\draw [fill=black] (4, 0) circle [radius=0.15];
\draw [ultra thick] (0, 0) -- (2, -2) -- (4, 0) -- (2, 2) -- (0, 0) -- (4, 0);

\node at (2, 0.3) {$500$};
\node at (0.7, 1.3) {$97$};
\node at (0.7, -1.3) {$53$};
\node at (3.3, 1.3) {$5$};
\node at (3.3, -1.3) {$83$};

\node at (2, -3) {\Huge $G$};

\draw [ultra thick, ->] (5, 0) -- (7, 0);

\begin{scope}[shift={(8, 0)}]
\draw [fill=black] (0, 0) circle [radius=0.15];
\draw [fill=black] (2, -2) circle [radius=0.15];
\draw [fill=black] (2, 2) circle [radius=0.15];
\draw [fill=black] (4, 0) circle [radius=0.15];
\draw [ultra thick] (0, 0) -- (2, -2) -- (4, 0) -- (2, 2) -- (0, 0) -- (4, 0);

\node at (2, -3) {\Huge $H$};

\node at (2, 0.3) {$4$};
\node at (0.7, 1.3) {$1$};
\node at (0.7, -1.3) {$2$};
\node at (3.3, 1.3) {$1$};
\node at (3.3, -1.3) {$1$};
\end{scope}
\end{tikzpicture}
\caption{\label{fig:shppres} Here $H$ has the same shortest paths as $G$, but its aspect ratio is improved from $100$ to $4$.}
\end{figure}
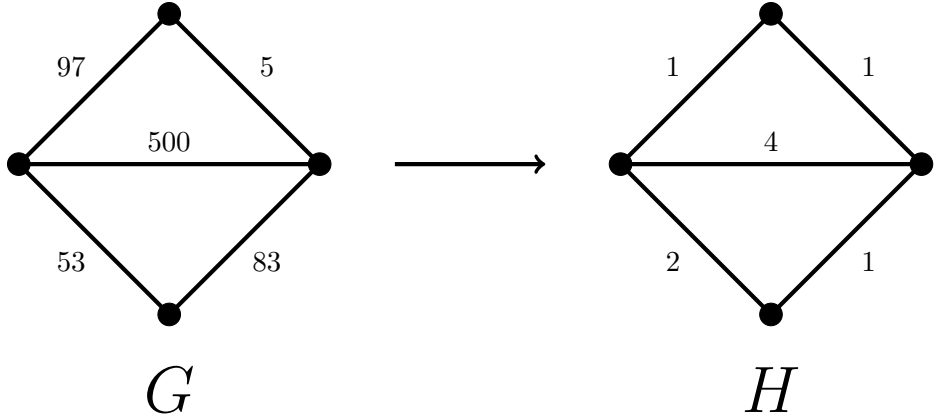

For an example of a potential use case for shortest-paths preservers, we would get a valid solution to APSP on a graph $G$ if we instead computed APSP over a shortest-paths preserving graph $H$, and then used the output as a solution for $G$ (and we could imagine leveraging the improved aspect ratio of $H$ to improve the computation).
We can restate our main questions formally using this definition:

\paragraph{Main Questions (Formal):}
\begin{itemize}
\item \emph{Does every $n$-node graph $G$ admit a shortest-paths preserving graph $H$ with aspect ratio $\poly(n)$?}

\item \emph{Are there $n$-node graphs $G$ whose shortest-paths system\footnote{Formally, the ``shortest-paths system" associated to a graph $G$ is the set of all vertex sequences that form shortest paths in $G$ \cite{CL22, CL23, Bodwin19}.} cannot be realized by any graph of aspect ratio $\poly(n)$?}
\end{itemize}

We can now formally point out that these questions are complements: a graph $G$ that resists a low aspect ratio shortest-paths preserver is precisely one that cannot be realized by a low aspect ratio graph.
Hence, such graphs (if they exist) imply that restricting to the setting of low aspect ratio confers additional structure on the shortest path systems to be analyzed.
Our choice to focus on aspect ratio $\poly(n)$ comes from the fact that, for many problems, the specific dependence on aspect ratio $W$ is $\polylog W$ (e.g. \cite{BW23, KP22a, RTZ08, ZL18, CDG20,CaoFR21,ForsterN18,HenzingerKN14,ChuzhoyZ23}), and so this translates to a $\polylog(n)$ dependence.

In addition to the potential algorithmic insights surveyed above, this paper fits into a recent line on the combinatorics of \emph{shortest path structures} \cite{AW20, AW23, CL22, CL23, Bodwin19}; it can be viewed as an investigation of the extremal shortest path systems that can only be induced by high aspect ratio graphs.
From this standpoint, an upper bound on aspect ratio would represent a new fact about the combinatorics of shortest path systems.
Meanwhile, lower bounds would imply that \emph{low-aspect-ratio} shortest path systems are fundamentally different objects than \emph{general} shortest path systems, meaning their combinatorial properties could be independently investigated.
We discuss this perspective in more depth in Section \ref{sec:shpstructure}.
More generally, there is virtually no current understanding of the interaction between shortest paths structure and aspect ratio, and the goal of this paper is to initiate this line of study. 


\subsection{Exact Shortest Path Preservers}

A priori, our main questions can be asked independently in different graph classes: e.g., we could envision settling the dichotomy differently for directed graphs, undirected graphs, DAGs, etc.
Our findings are that this is in fact the case: the answer does in fact change between these three graph classes.


Let us begin our discussion in the setting of DAGs.
As a warmup observation, there is a simple way to construct $n$-node graphs that require aspect ratio $\Omega(n)$ for any shortest-paths preserving graph (see Figure \ref{fig:intropathlb}). The lower bound is a path of $n$ nodes and $n-1$ edges of weight $1$, plus an additional long edge of weight $n$ that jumps from the start to the end of the path.
Since the path of unit-weight edges is a shortest path, in reweighting, one cannot reduce the weight of the long edge below $n-1$ times the minimum edge weight on the path.
This implies that any shortest-paths preserving graph has aspect ratio $\Omega(n)$.

\usetikzlibrary{calc}
\begin{figure}[h]
\begin{center}
\begin{tikzpicture}
    \newcommand{\numnodes}{7}

    \fill (0, 0) circle (2pt) node (start) {};
    \foreach \i in {0,...,\numnodes} {
        \node at (\i, 0) (node\i) {};
        \fill (node\i) circle (2pt); 
    }

    \foreach \i in {1,...,\numnodes} {
        \ifnum\i=4 
            \node at ($(node\the\numexpr\i-1\relax)!0.5!(node\i)$) {$\cdots$}; 
        \else
            \draw[->] (node\the\numexpr\i-1\relax) -- node[midway, below] {$1$} (node\i);
        \fi
    }

    \draw [->, bend left] (node0) to node[midway, above] {$n$} (node\numnodes);
\end{tikzpicture}

\end{center}
\caption{\label{fig:intropathlb} Any shortest-paths preserving reweighting of this graph has aspect ratio $\Omega(n)$.}
\end{figure}
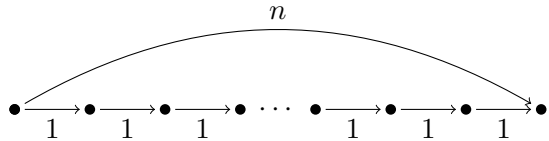

Our first result is that this simple lower bound is actually tight, over the class of DAGs:

\begin{restatable}[Linear upper bound for DAGs]{theorem}{ub}\label{thm:ub}
    Every $n$-node DAG has a shortest-paths preserving graph of aspect ratio $O(n)$.
\end{restatable}


The fact that nontrivial aspect ratio upper bounds are possible might give hope that they generalize to an upper bound of $O(n)$, or perhaps $\poly(n)$, for \emph{all} directed graphs.
However, our next result refutes this possibility: DAGs are special, and in general \emph{exponential} aspect ratio is sometimes necessary. That is, despite the fact that only $n^2$ shortest paths need to be preserved, we discover that the interactions between these paths can compound on one another to force a chain reaction of edges with multiplicatively larger and larger weights. The consequence is a graph with ``inherently'' exponential aspect ratio.

\begin{restatable}[Exponential lower bound for directed graphs]{theorem}{lbdir} \label{thm:lbdir}
    There are $n$-node directed graphs $G$ such that any shortest-paths preserving graph $H$ has aspect ratio $2^{\Omega(n)}$.
\end{restatable}

Next, we consider the undirected setting.
The undirected setting is at least as easy as the directed setting: there is a black-box reduction showing that, if every directed graph has a shortest-paths preserver of aspect ratio at most $\alpha(n)$, then the same is true for undirected graphs \cite{Bodwin19}.\footnote{To sketch the reduction: given an undirected graph $G$, convert it to a directed graph $G'$ by replacing each undirected edge $\{u, v\}$ with both directed edges $(u, v), (v, u)$.  Let $H'$ be a shortest-paths preserving graph of $G'$ with aspect ratio $\alpha(n)$.  Convert $H'$ to an undirected graph $H$ by recombining edges $(u, v), (v, u)$ into a single undirected edge $\{u, v\}$ of weight $w_{H'}(u, v) + w_{H'}(v, u)$. 
 One can calculate that this recombination step does not increase aspect ratio, and it preserves shortest paths.} 
Nonetheless, we extend our lower bound to show that exponential aspect ratio is necessary for general undirected graphs as well.

\begin{restatable}[Exponential lower bound for undirected graphs]{theorem}{lbundir} \label{thm:lbundir}
    There are $n$-node undirected graphs $G$ such that any shortest-paths preserving graph $H$ has aspect ratio $2^{\Omega(n)}$.
\end{restatable}

\subsection{Shortest Path Preservers with Stretch}

Given our hardness results for general graphs, it is next natural to investigate whether they can be overcome by allowing an approximation factor.
We will consider the following model:
\begin{definition}[$\alpha$-stretch shortest-paths preservering graph]
    Given a graph $G=(V,E,w_G)$, a reweighted graph $H=(V,E,w_H)$ is \emph{$\alpha$-stretch shortest-paths preserving} if every shortest path in $H$ is also an $\alpha$-approximate shortest path in $G$.
    That is, if $\pi$ is an $s \leadsto t$ shortest path in $H$, then $w_G(\pi) \le \alpha \cdot \dist_G(s, t).$
\end{definition}

This definition is generally weaker than the (exact) shortest-paths preserving graphs discussed previously, which is the special case of stretch $1$.
It represents a natural attempt to bypass the lower bounds in Theorems \ref{thm:lbdir} and \ref{thm:lbundir}: can we reduce aspect ratio to $\poly(n)$ in a more forgiving model, where we allow a stretch factor $\alpha$ in the shortest paths of $G$?

We describe modifications of our lower bounds that are robust even to this approximate version of the problem.
In the directed setting, we construct graphs that force exponential aspect ratio for \emph{any} finite $\alpha$ (even depending on $n$):

\begin{restatable}[Exponential lower bound for approximation in directed graphs]{theorem}{lbdirap}\label{thm:lbdirap}
   For any $\alpha$, there are $n$-node directed graphs $G$ such that any $\alpha$-stretch shortest-paths preserving graph $H$ has aspect ratio $2^{\Omega(n)}$.
\end{restatable}

For undirected graphs, we are able to modify our construction to obtain an exponential lower bound for the $(1+\eps)$-approximate version:

\begin{restatable}[Exponential lower bound for approximation in undirected graphs]{theorem}{lbundirap} \label{thm:lbundirap}
    For any $\eps \le 1/12$, there are $n$-node undirected graphs $G$ such that any $(1+\eps)$-stretch shortest-paths preserving graph $H$ has aspect ratio $2^{\Omega(n)}$.
\end{restatable}

We have not attempted to optimize the constant $1/12$ in this theorem: it is almost certainly improvable, but only to a point (e.g.\ our methods will likely not be able to rule out shortest-paths preserving graphs with $\text{poly}(n)$ aspect ratio and stretch $2$).
It is an interesting open question whether one can obtain $\text{poly}(n)$ aspect ratio and $O(1)$ stretch in the undirected setting.

\subsection{Shortest Path Preservers with Two-Sided Stretch}

In the previous settings, we imagine computing \emph{exact} shortest paths over $H$, and then mapping them to exact or approximate shortest paths in $G$.
One might also imagine computing \emph{approximate} shortest paths in $H$, and designing $H$ such that these map to approximate shortest paths in $G$.
The following definition captures the property that $H$ would need to have to enable this method.

\begin{definition}[$(\alpha_H \to \alpha_G)$-stretch shortest-paths preservering graph]
    Given a graph $G=(V,E,w)$, a reweighted graph $H=(V,E,w_H)$ is said to be \emph{$(\alpha_H \to \alpha_G)$-stretch shortest-paths preserving} if every $\alpha_H$-approximate shortest path in $H$ is also an $\alpha_G$-approximate shortest path in $G$.
\end{definition}

Thus, an $\alpha$-stretch preserver in the previous sense is the same as a $(1 \to \alpha)$-stretch preserver under this definition.
We note that the problem is easier (more likely to admit an upper bound) as $\alpha_G$ increases, but harder as $\alpha_H$ increases.
Thus, the previous lower bounds for directed and undirected graphs from Theorems \ref{thm:lbdirap} and \ref{thm:lbundirap} apply also to $(\alpha_H \to \alpha)$-stretch preservers, for any $\alpha_H \ge 1$.
However, the two-sided stretch version is generally incomparable in difficulty to the original (exact shortest paths) version of the problem.\footnote{On one hand, the exact version is harder because an exact shortest path in $G$ must be an exact shortest path in $H$, whereas in the two-sided stretch version, an exact shortest path in $G$ only needs to be an approximate shortest path in $H$.
On the other hand, the two-sided stretch version is harder because \emph{every} approximate shortest path in $H$ must also be an approximate shortest path in $G$, whereas in the exact version there are no constraints on the approximate-but-not-exact shortest paths of $H$.}
Thus, we can study this two-sided stretch question for DAGs: Do DAGs have $(\alpha_H \to \alpha_G)$-stretch preservers with polynomial aspect ratio, for $\alpha_H, \alpha_G > 1$?
Our final result is an answer to this question in the negative: 

\begin{restatable}[Exponential lower bound for two-sided approximation in DAGs]{theorem}{lbdagap} \label{thm:lbdagap}
    For any $\alpha_H,\alpha_G > 1$, there is an $n$-node DAG $G$ such that any $(\alpha_H \to \alpha_G)$-stretch shortest-paths preserving graph $H$ has aspect ratio $(\alpha_H)^{\Omega(\sqrt{n})}$.
\end{restatable}

This implies a dichotomy: for the exact version on DAGs the answer is $O(n)$, while for the two-sided approximate version, the answer suddenly jumps to exponential.

\subsection{\label{sec:shpstructure} Discussion: Aspect Ratio and Theory of Shortest Paths Structure}

As mentioned previously, this work can be viewed as a study of the combinatorics of \emph{shortest path systems}: we construct extremal shortest path systems that can only be expressed using large aspect ratio, or (in the case of DAGs) we prove that no such structures exist.
In this sense, our work fits into a recent line of research in theoretical computer science and combinatorics that centers shortest path systems as objects exhibiting remarkable combinatorial structure.
Some other work of this kind includes:

\begin{itemize}
\item Amiri and Wargalla \cite{AW20} proved that, for DAGs $G$, it holds that every triple of nodes lies along a shortest path iff there is a \emph{single} shortest path that covers all nodes in the graph.
Extensions of this theorem to undirected and directed graphs were later established by Akmal and Wein \cite{AW23}.

\item A developing line of work, initiated by Cizma and Linial \cite{CL22, CL23}, studies the class of \emph{geodesic graphs}, which are the graphs in which every complete consistent path complex\footnote{Formally, a path system is \emph{complete} if it contains one path between each node pair, and it is \emph{consistent} if it is closed under taking contiguous subpaths.} can be induced as shortest paths by an edge weight function.
Many graphs satisfying these condition have been discovered, as have obstructions to geodesy.

\item Bodwin \cite{Bodwin19} classifies the combinatorial patterns that can or can't generally appear in the shortest path system induced by a graph with unique shortest paths.
\end{itemize}

Other results along these lines include \cite{Balzotti22, CPFGHKS20}.
All of the above results study the combinatorial structure of shortest path systems that are induced by graphs with \emph{arbitrary} weight functions.
Our results suggest a way to incorporate edge weights into the discussion.
In particular, for general directed and undirected graphs, our results show that the class of shortest paths systems induced by \emph{low aspect-ratio} weight functions is a proper subset of the class of shortest path systems induced by \emph{any} weight function.
This therefore raises the open problem of how the structural results from prior work evolve if we restrict our attention to graphs with ``simple'' edge weights, in the sense of small aspect ratio.
(Some other restrictions on the weight function could be interesting as well.)


\subsection{Open Problems and Future Directions}

Besides the upper bound for DAGs (Theorem \ref{thm:ub}), the findings of this paper point towards the message that \emph{low aspect ratio graphs have additional structure in their shortest paths}.
We conclude our introduction by mentioning three places in which there is still potential to settle our main dichotomy in the other direction.
First: although we have placed focus on directed graphs, undirected graphs, and DAGs, there are many other notable graph classes in which one could seek upper bounds: 

\begin{oq}
Are there other notable graph classes that always admit shortest-paths preserving graphs of aspect ratio $\poly(n)$, besides DAGs?
\end{oq}

Our second problem was mentioned previously, in the context of shortest-paths preservers with stretch for undirected graphs.
Although we have constructed graphs for which there is no $(1+\eps)$-stretch preserver of aspect ratio $\poly(n)$, such a preserver of stretch $O(1)$ is still conceivable:
\begin{oq}
Does every $n$-node undirected graph have a constant-stretch shortest-paths preserving graph of aspect ratio $\poly(n)$?
\end{oq}

Although we focus on aspect ratio in this paper, the same basic questions apply to other complexity measures associated to graph edge weights.
Perhaps most naturally, one can re-ask our questions for graphs of integer edge weights, parametrized by maximum edge weight:
\begin{oq}
Does every $n$-node DAG have a shortest-paths preserving graph $H$ with \textbf{integer} edge weights in the range $[1, \dots, \poly(n)]$?
If not, does this hold if we allow stretch $\alpha$?
\end{oq}

We note that this question is harder (less likely to admit an upper bound) than the one for aspect ratio, since any graph with integer edge weights in the range $[1, \dots, \poly(n)]$ has aspect ratio $\poly(n)$.
Thus, the answer to the question is \emph{no} in the settings of general directed or undirected graphs. But an affirmative answer for DAGs is still possible, and would for example allow us to extend the parallel/distributed results in Corollary 1.8 and 1.9 of \cite{RozhonHMGZ23} to work for arbitrary weighted graphs (assuming $H$ could be computed efficiently).

Finally: this paper focuses on \emph{extremal} questions, investigating the extent to which aspect ratio can be reduced for \emph{all} graphs.
One could also study an instance-optimal version of the problem, in which the goal is to compute a shortest-paths preserving graph of minimum aspect ratio:
\begin{oq}
Is there a polynomial-time algorithm that takes an input graph $G$ on input and computes a shortest-paths preserving graph (possibly with stretch) of minimum or near-minimum aspect ratio?
If not, does this problem admit an approximation algorithm?
\end{oq}


\subsection{Paper Organization}

In \cref{sec:dir} we prove our exponential lower bound for directed graphs, which serves as a warm-up for our exponential lower bound for undirected graphs, which appears in \cref{sec:undir}. In \cref{sec:dir,sec:undir} we also include the extensions of our exponential lower bounds to the approximate $\alpha$-stretch version of the problem for directed and undirected graphs respectively. In \cref{sec:dagub} we prove our $O(n)$ upper bound for DAGs. Finally, in \cref{sec:lbdag} we prove our exponential lower bound for the two-sided stretch version of the problem for DAGs. 

Instead of a centralized technical overview, we include a high-level overview of our approach at the beginning of several of the individual sections.

\section{Warm Up: Exponential Lower Bounds for Directed Graphs}\label{sec:dir}

We start by proving our lower bounds for general directed graphs, as these are simpler. We begin with the exact version, and later move to the approximate version. In the next section, we turn to undirected graphs. We first prove Theorem \ref{thm:lbdir}, restated below.

\lbdir*

\paragraph{High-Level Approach:} It is clear from the example in \cref{fig:intropathlb} that one can easily find graphs where any shortest-paths preserving reweighting has aspect ratio $\Omega(n)$. This example consists of a shortest path $P$ of weight $n$, and a not-shortest path $P'$ (in this case just a single edge) with the same endpoints as $P$. The fact that $P'$ must remain not shortest after reweighting provides a lower bound on the sum of the weights of the edge(s) in $P'$. 

The key question towards getting a lower bound with larger aspect ratio is: can this approach be iterated? That is, we would like to carry out the following procedure:
\begin{enumerate}
\item Generate a collection of edges of at least some weight $w$. 
\item Combine these edges into a collection $\mathcal{P}$ of shortest paths where each path in $\mathcal{P}$ has $\ell$ edges of weight $w$, for some $\ell$.
\item Construct a collection $\mathcal{P}'$ of \emph{not} shortest paths with the same endpoints as the shortest paths in $\mathcal{P}$, such that each of the paths in $\mathcal{P}'$ has $\ell'<\ell$ edges.
\item Now, in order for the paths in $\mathcal{P}'$ to remain not-shortest after reweighting, their edge weights must be on average a multiplicative factor of $\ell/\ell'$ larger than the weight $w$ of the edges in $\mathcal{P}$. Now that we have generated a collection of edges of weight multiplicatively larger than $w$, we can return back to step 1 with a larger value of $w$.
\end{enumerate}

Several challenges arise when trying to construct a graph with these properties: 

First, an edge constructed in some iteration $i$ of this procedure has a very restricted set of properties: it must simultaneously be in a \emph{not-shortest} path with the \emph{same} endpoints as a shortest path from iteration $i$, and be in a \emph{shortest} path in the next iteration $i+1$. This requires a precise interleaving of various shortest and not-shortest paths, and it is unclear whether such a construction should exist. 

Second, it is easy to imagine that if one uses $k$ edges of weight $w$ in the infrastructure of iteration $i$, one might only generate a smaller number, say $k/2$, edges of larger weight for the next iteration. This type of situation would not yield strong bounds because it would only allow $\log n$ iterations of the procedure, which would result in only polynomial aspect ratio. Thus, we need to make sure there is not too much loss  in the number of edges we generate in each iteration (and in fact we achieve no loss).

We show how to construct a surprisingly simple graph that satisfies the constraints necessary to carry out the iterated procedure, and furthermore uses only 3 edges of average weight $w$ (plus 3 helper edges) to generate 3 more edges whose average weight is at least $2w$. Because of this, we can perform $\Omega(n)$ iterations of the procedure, which yields an aspect ratio of $2^{\Omega(n)}$.

Now we will provide our construction and analysis in detail.

\subsection{Construction}

\begin{figure}[t]
\begin{center}
\begin{tikzpicture}

\def\dist{2.5}

\def\crosscolor{black}

\node [style={circle,fill=blue!20}] (A1) at (0,0) {$v_1^1$};
\node [style={circle,fill=blue!20}] (A2) at (1,1) {$v_1^2$};
\node [style={circle,fill=blue!20}] (A3) at (2,0) {$v_1^3$};
\draw[->] (A1) -- (A2);
\draw[->] (A2) -- (A3);
\draw[->] (A3) -- (A1);
\node [align=center, draw=none, fill=none] at (1, -1.5) {forward cycle,\\edge wts $1/3$};

\node [style={circle,fill=blue!20}] (B1) at (2+\dist,0) {$v_2^1$};
\node [style={circle,fill=blue!20}] (B2) at (3+\dist,1) {$v_2^2$};
\node [style={circle,fill=blue!20}] (B3) at (4+\dist,0) {$v_2^3$};
\draw[->] (B2) -- (B1);
\draw[->] (B3) -- (B2);
\draw[->] (B1) -- (B3);
\node [align=center, draw=none, fill=none] at (3+\dist, -1.5) {backward cycle,\\edge wts $1/9$};

\node [style={circle,fill=blue!20}] (C1) at (4+2*\dist,0) {$v_3^1$};
\node [style={circle,fill=blue!20}] (C2) at (5+2*\dist,1) {$v_3^2$};
\node [style={circle,fill=blue!20}] (C3) at (6+2*\dist,0) {$v_3^3$};
\draw[->] (C1) -- (C2);
\draw[->] (C2) -- (C3);
\draw[->] (C3) -- (C1);
\node [align=center, draw=none, fill=none] at (5+2*\dist, -1.5) {forward cycle,\\edge wts $1/27$};

\draw [->, \crosscolor, dotted] (A2) to[bend left=20] (B2);
\draw [->, \crosscolor, dotted] (B2) to[bend left=20] (C2);
\draw [->, \crosscolor, dotted] (A1) to[bend right=20] (B1);
\draw [->, \crosscolor, dotted] (B1) to[bend right=20] (C1);
\draw [->, \crosscolor, dotted] (A3) to[bend right=20] (B3);
\draw [->, \crosscolor, dotted] (B3) to[bend right=20] (C3);

\node [right=1 of C3] {\Huge \bf $\cdots$};

\end{tikzpicture}
\caption{\label{fig:dirconstruction}
Our exponential aspect ratio lower bound for directed graphs (Theorem \ref{thm:lbdir}). The graph consists of $n/3$ cycles of length 3. The cross-cycle edges (gray, dotted) go between same-numbered nodes in adjacent cycles, and all have weight $0$.}
\vspace{9mm}

\begin{tikzpicture}
\def\dist{2.5}
\def\crosscolor{black}

\node [style={circle,fill=blue!20,minimum width=1.2cm}] (A1) at (0,0) {$v_i^1$};
\node [style={circle,fill=blue!20,minimum width=1.2cm}] (A2) at (1,1.5) {$v_i^2$};
\node [style={circle,fill=blue!20,minimum width=1.2cm}] (A3) at (2,0) {$v_i^3$};
\draw[->] (A1) -- (A2);
\draw[->, ultra thick, red] (A2) -- (A3);
\draw[->] (A3) -- (A1);
\node [align=center, draw=none, fill=none] at (1, -1.5) {forward cycle,\\edge wts $1/3^i$};

\node [style={circle,fill=blue!20,minimum width=1.2cm}] (B1) at (2+\dist,0) {$v_{i+1}^1$};
\node [style={circle,fill=blue!20,minimum width=1.2cm}] (B2) at (3+\dist,1.5) {$v_{i+1}^2$};
\node [style={circle,fill=blue!20,minimum width=1.2cm}] (B3) at (4+\dist,0) {$v_{i+1}^3$};
\draw[->, blue, ultra thick] (B2) -- (B1);
\draw[->] (B3) -- (B2);
\draw[->, blue, ultra thick] (B1) -- (B3);
\node [align=center, draw=none, fill=none] at (3+\dist, -1.5) {backward cycle,\\edge wts $1/3^{i+1}$};

\draw [->, blue, dotted, ultra thick] (A2) to[bend left=20] (B2);
\draw [->, \crosscolor, dotted] (A1) to[bend right=20] (B1);
\draw [->, red, dotted, ultra thick] (A3) to[bend right=20] (B3);

\end{tikzpicture}

\caption{\label{fig:dirpaths} Our analysis uses the fact that the shortest path between certain node pairs in adjacent cycles uses no edges from $C_i$ and $2$ edges from $C_{i+1}$ (the blue path -- see Lemma \ref{lem:spdir}), rather than an alternate non-shortest path that uses $1$ edge $C_i$ and  no edges from $C_{i+1}$ (the red path).  For example, in this diagram, the shortest $v_i^2 \leadsto v_{i+1}^3$ path is $(v_i^2, v_{i+1}^2, v_{i+1}^1, v_{i+1}^3)$ (thick, blue) and not $(v_i^2, v_i^3, v_{i+1}^3)$ (thick, red). This ultimately implies that, in any shortest-paths preserving graph, the blue path must remain shorter than the red path and so edge weights in $C_{i}$ must be (on average) at least double of those in $C_{i+1}$ (see Lemma \ref{lem:aspectdir})}.

\end{center}
\end{figure}

We construct the lower-bound graph $G = (V,E,w)$ as follows (see Figure \ref{fig:dirconstruction})

\begin{itemize}
\item We assume that the number of vertices $n$ is divisible by $3$. The graph $G = (V,E,w)$ consists of $n/3$ cycles, each with exactly $3$ vertices. Label the cycles $C_1, ... C_{n/3}$, where cycle $C_i$ consists of vertices $v^1_i, v^2_i, v^3_i$.
\item Each cycle is either a \emph{forward cycle} or a \emph{backward cycle}. If $C_i$ is a forward cycle, then the graph contains directed edges $(v^1_i, v^2_i)$, $(v^2_i, v^3_i)$ and $(v^3_i, v^1_i)$. If $C_i$ is a backward cycle then the graph contains directed edges $(v^1_i, v^3_i)$, $(v^3_i, v^2_i)$ and $(v^2_i, v^1_i)$. The cycles alternate forward and backward. That is, $C_i$ is forward for odd $i$ and backward for even $i$.
\item There are also edges between cycles, which we call \emph{cross-cycle edges}. In particular, for every $i$, there is an edge $(v^1_i, v^1_{i+1})$, as well as edges $(v^2_i, v^2_{i+1})$ and $(v^3_i, v^3_{i+1})$. Note that edges only go \textit{from} $C_i$ \textit{to} $C_{i+1}$; there are no edges from $C_i$ to any other $C_j$, and in particular there are no edges going from $C_i$ to $C_{i-1}$. (That is, the graph is almost a DAG, except that the $C_i$ themselves are cycles.)
\item For any edge $e \in C_i$ we set $w(e) = 1/3^i$. We set the weight of all cross-cycle edges to $0$\footnote{If we want to keep edge-weights positive, then the same proof goes through if we set $w(e) = \delta$ for all cross-cycle edges; for the exact lower bound (Theorem \ref{thm:lbdir}) the value of $\delta$ does not matter, but for the approximate lower bound (Theorem \ref{thm:lbdirap}) we would need to set $\delta$ to be very tiny.}

\end{itemize}

\subsection{Analysis}
\label{sec:dir-analysis}

The analysis is broken into two parts. First, we will identify a particular set $\mathcal{P}$ of paths and in \cref{lem:spdir} we will show that they are unique shortest paths in $G$ under weight function $w$. Then, in \cref{lem:aspectdir} we will show that any new weight function $w'$ on $G$ for which all paths in $\mathcal{P}$ remain shortest must have aspect ratio $2^{\Omega(n)}$.

\paragraph{Construction of $\mathcal{P}$.} See the blue path in Figure \ref{fig:dirpaths} for an example of a path in $\mathcal{P}$. All paths in $\mathcal{P}$ go from a vertex in $C_i$ to a vertex in $C_{i+1}$. Consider some cycle $C_i$. Let $v^k_i, v^{k'}_i$ be any two consecutive vertices in $C_i$: so if $C_i$ goes forward then we have $(k,k') \in \{ (1,2), (2,3), (3,1)\}$, while if $C_i$ goes backward then we have  $(k,k') \in \{ (1,3), (3,2), (2,1)\}$. 

For any such consecutive pair $v^k_i, v^{k'}_i$, we add the following path from $v^k_i$ to $v^{k'}_{i+1}$ to $\mathcal{P}$: the path takes a cross-cycle edge from $v^k_i$ to $v^{k}_{i+1}$, and then it follows $C_{i+1}$ from $v^{k}_{i+1}$ to $v^{k'}_{i+1}$. For example, if $C_i$ is a forward cycle and the consecutive pair is $v^2_i, v^3_i$, then we add to $\mathcal{P}$ the following 3-edge path from $v^2_i$ to $v^3_{i+1}$: $(v^2_{i}, v^2_{i+1}) \circ (v^2_{i+1}, v^1_{i+1}) \circ (v^1_{i+1}, v^3_{i+1})$. 

\begin{lemma}\label{lem:spdir}
Each path in $\mathcal{P}$ is the shortest path between its endpoints in $G$. 
\end{lemma}
\begin{proof} Fix a path $P\in\mathcal{P}$, where $P$ is from $v^k_i$ to $v^{k'}_{i+1}$ for some consecutive pair $v^k_i, v^{k'}_i$. It is easy to see that $w(P) = 2/3^{i+1}$, which corresponds to two edges on $C_{i+1}$ (recall that the cross-cycle edge has weight $0$). Let $P'$ be an alternate path with the same endpoints as $P$. We will argue that $w(P')>w(P)$. Each edge in $C_i$ is of weight at least $1/3^i$, so if $P'$ uses an edge in $C_i$ then indeed $w(P') \geq 1/3^i > w(P)$. The only edge from $v^k_i$ that is not in $C_i$ is the first edge of $P$: $(v^k_i, v^{k}_{i+1})$, so $P'$ must begin with this edge. From $v^{k}_{i+1}$, if $P'$ uses another cross-cycle edge, it can never return to $C_{i+1}$ since all cross-cycle edges go from some $C_j$ to $C_{j+1}$. Therefore, the only edges that $P'$ can take from $v^{k}_{i+1}$ are edges in $C_{i+1}$. Because $C_{i+1}$ is simply a directed cycle, there is only one path to $v^{k'}_{i+1}$, which is precisely the path that $P$ takes.
\end{proof}

\begin{lemma}\label{lem:aspectdir}
Any new weight function $w_H$ on $G$ for which all paths in $\mathcal{P}$ remain shortest has aspect ratio $2^{\Omega(n)}$.
\end{lemma}

\begin{proof}
    Let $H=(V,E,w_H)$ be a shortest-paths-preserving graph of $G$. We will show that the weight of each cycle in $H$ is double that of its successor; that is, for all $i\in \{1,\dots,n/3-1\}$, we will show that $w_H(C_i)> 2\cdot w_H(C_{i+1})$. This implies that the aspect ratio of $w_H$ is $2^{\Omega(n)}$.
    
    For any consecutive pair $v^k_{i}$, $v^{k'}_{i}$, we know that in $H$ the 3-edge path in $\mathcal{P}$ from $v^k_i$ to $v^{k'}_{i+1}$ is shorter than the 2-edge path $(v^k_{i}, v^{k'}_{i}) \circ (v^{k'}_{i}, v^{k'}_{i+1})$. (These correspond to the blue and red paths respectively in Figure \ref{fig:dirpaths}.) Taking the sum of this inequality for all three consecutive pairs in $C_i$, we get $w_H(C_i) + X > 2w_H(C_{i+i}) + X$, where $X = \sum_{k=1}^3 w_H(v^k_i, v^k_{i+1})$ is the sum (in $H$) of the weights of the cross-cycle edges from $C_i$ to $C_{i+1}$. The $X$ cancels, so we get $w_H(C_i) > 2w_H(C_{i+i})$, as desired. 
\end{proof}

\subsection{Approximate Version}

We now show that in directed graphs, some graphs require exponential aspect ratio even if we allow arbitrary stretch. We prove Theorem \ref{thm:lbdirap}, restated below:

\lbdirap*

\paragraph{Construction:} the construction of $G$ is exactly the same as for the exact version above, except all edges in $C_i$ have weight $1/(3\cdot \alpha)^i$ instead of $1/3^i$. 

\paragraph{Analysis:} we will use the same set $\mathcal{P}$ of paths as for the exact version. We will show in \cref{lem:spdirap} that each path in $\mathcal{P}$ is not only the shortest path between its endpoints in $G$, but is also the \emph{only} $\alpha$-approximate shortest path between its endpoints in $G$. Since $H$ is an $\alpha$-stretch shortest-path preserving graph of $G$, this implies that each path in $\mathcal{P}$ is the \emph{unique} shortest path between its endpoints in $H$. We have already shown in \cref{lem:aspectdir} that any new weight-function $w_H$ on $G$ for which all paths in $\mathcal{P}$ remain shortest paths has aspect ratio $2^{\Omega(n)}$. Thus, it only remains to prove \cref{lem:spdirap}:

\begin{lemma}\label{lem:spdirap}
    Each path in $\mathcal{P}$ is the only $\alpha$-approximate shortest path between its endpoints in $G$.
\end{lemma}
\begin{proof}
   Fix a path $P\in\mathcal{P}$. We know that $P$ is a path from $v^k_i$ to $v^{k'}_{i+1}$ for some consecutive pair $v^k_i, v^{k'}_i$. It is easy to see that $w(P) = \frac{2}{(3\cdot \alpha)^{i+1}}$. Let $P'$ be an alternate path with the same endpoints as $P$. We will argue that $w(P')>\alpha\cdot w(P)$. 
   
Each edge in $C_i$ is of weight at least: 
$$\frac{1}{(3\alpha)^i}>\alpha\cdot \frac{2}{(3\cdot\alpha)^{i+1}} =\alpha\cdot w(P),$$
so if $P'$ uses an edge in $C_i$ then indeed $w(P')>\alpha\cdot w(P)$. The only edge from $v^k_i$ that is not in $C_i$ is the first edge of $P$: $(v^k_i, v^{k}_{i+1})$, so $P'$ must begin with this edge. From $v^{k}_{i+1}$, if $P'$ uses another cross-cycle edge, it can never return to $C_{i+1}$ since all cross-cycle edges go from some $C_j$ to $C_{j+1}$. Therefore, the only edges that $P'$ can take from $v^{k}_{i+1}$ are edges in $C_{i+1}$. Because $C_{i+1}$ is simply a directed cycle, there is only one path to $v^{k'}_{i+1}$, which is precisely the path that $P$ takes.
\end{proof}

\section{Exponential Lower Bounds for Undirected Graphs}\label{sec:undir}
In this section, we prove our lower bounds for undirected graphs, which are somewhat more complicated. We start with the exact shortest-path preservers (theorem restated below), and turn to approximate ones in the next subsection. 
\lbundir*

\paragraph{High-Level Approach:} Our construction has the same basic structure as the lower bound for directed graphs in the previous section. The graph $G$ will again consists of constant-length cycles $C_1$, .., $C_k$, and we will again construct the graph in such a way as to ensure that the shortest-path-preserver $H$ must satisfy $w_H(C_i) > 2 \cdot w_H(C_{i+1})$; since the number of cycles is $\Omega(n)$, this implies that $w_H$ has aspect ratio $2^{\Omega(n)}$. 

At a high level, the argument in directed graphs relied on the fact that every path $P \in \mathcal{P}$ from $C_i$ to $C_{i+1}$ consists of $2$ edges in $C_{i+1}$ (plus a cycle-crossing edge), yet is shorter than an alternative path $P'$ between the same endpoints which consists of $1$ edge in $C_i$ (plus a cycle-crossing edge) -- see Figure \ref{fig:dirpaths}. This implies that edges in $C_i$ must have at least double the weights of those in $C_{i+1}$, as desired. 

In order to construct an undirected graph which provides a similar guarantee, we need to overcome three new issues with this approach that are specific to undirected graphs. 

The first issue is that in directed graphs we ensured that $P$ has fewer edges than $P'$ by alternating the direction of the cycles. But in undirected graphs, there can be no ``forward" or ``backward" cycles. We overcome this issue by introducing more structure to the cross-cycle edges, which ensures that every $P \in \mathcal{P}$ still uses two edges from $C_{i+1}$, and that there is still an alternative path $P'$ using a single edge from $C_i$. In particular, the cross-cycle edges no longer simply go from each vertex in $C_i$ to its copy in $C_{i+1}$.

The second issue is that we need our construction to control the direction that the shortest path $P$ will follow along the cycle. This is easily accomplished by increasing the size of each cycle to 5, so that following the $2$ edges of $C_{i+1}$ on $P$ is shorter than following the $3$ edges of $C_{i+1}$ in the other direction.

Finally, the third issue is that in undirected graphs, there can also be paths from $C_i$ to $C_{i+1}$ that go through multiple levels of cycles (e.g. to $C_{i+2}$) and then return back to $C_{i+1}$. (In the directed construction, this could not happen because edges only pointed from lower to higher numbered cycles.) Because of this, we no longer have the freedom to set the cross-cycle edges to have weight $0$. In fact, if they were weight 0 then there would be a path of weight 0 from every vertex to every other vertex in the graph, due to the newly defined cross-cycle edges from the first issue. Instead, we need to set the weights of the cross-cycle edges to be higher; for the exact lower bound we can simply set the cross-cycle weights to be very large, but for the approximate lower bound we need to balance them with the weights of the cycle edges. See the high-level description within the approximate version section (\cref{sec:undirapprox}) for more details on this issue.

\subsection{Construction}
We construct our lower-bound graph $G = (V,E,w)$ as follows (see Figure \ref{fig:undirconstruction}):

\begin{figure}[ht]
\begin{center}
\begin{tikzpicture}

\def\dist{4.5}
\def\crosscolor{black}

\node [style={circle,fill=blue!20}] (A1) at (72*1:1) {$v_1^1$};
\node [style={circle,fill=blue!20}] (A2) at (72*2:1) {$v_1^2$};
\node [style={circle,fill=blue!20}] (A3) at (72*3:1) {$v_1^3$};
\node [style={circle,fill=blue!20}] (A4) at (72*4:1) {$v_1^4$};
\node [style={circle,fill=blue!20}] (A5) at (72*5:1) {$v_1^5$};
\draw (A2) -- (A1);
\draw (A3) -- (A2);
\draw (A4) -- (A3);
\draw (A5) -- (A4);
\draw (A1) -- (A5);

\node at (0, -2) {edge wts $1/3$};

\node [style={circle,fill=blue!20}] (B1) [right=5 of A1] {$v_2^1$};
\node [style={circle,fill=blue!20}] (B2) [right=5 of A2] {$v_2^2$};
\node [style={circle,fill=blue!20}] (B3) [right=5 of A3] {$v_2^3$};
\node [style={circle,fill=blue!20}] (B4) [right=5 of A4] {$v_2^4$};
\node [style={circle,fill=blue!20}] (B5) [right=5 of A5] {$v_2^5$};
\draw (B1) -- (B2);
\draw (B2) -- (B3);
\draw (B3) -- (B4);
\draw (B4) -- (B5);
\draw (B5) -- (B1);
\node at (6, -2) {edge wts $1/9$};

\node [style={circle,fill=blue!20}] (C1) [right=5 of B1] {$v_3^1$};
\node [style={circle,fill=blue!20}] (C2) [right=5 of B2] {$v_3^2$};
\node [style={circle,fill=blue!20}] (C3) [right=5 of B3] {$v_3^3$};
\node [style={circle,fill=blue!20}] (C4) [right=5 of B4] {$v_3^4$};
\node [style={circle,fill=blue!20}] (C5) [right=5 of B5] {$v_3^5$};
\draw (C2) -- (C1);
\draw (C3) -- (C2);
\draw (C4) -- (C3);
\draw (C5) -- (C4);
\draw (C1) -- (C5);
\node at (12, -2) {edge wts $1/27$};

\draw[\crosscolor,dotted] (A1) to[bend right=20] (B1);
\draw[\crosscolor,dotted] (A2) to[bend right=20] (B3);
\draw[\crosscolor,dotted] (A3) to[bend right=20] (B5);
\draw[\crosscolor,dotted] (A4) to[bend right=20] (B2);
\draw[\crosscolor,dotted] (A5) to[bend right=20] (B4);

\draw[\crosscolor,dotted] (B1) to[bend right=20] (C1);
\draw[\crosscolor,dotted] (B2) to[bend right=20] (C3);
\draw[\crosscolor,dotted] (B3) to[bend right=20] (C5);
\draw[\crosscolor,dotted] (B4) to[bend right=20] (C2);
\draw[\crosscolor,dotted] (B5) to[bend right=20] (C4);

\node [right=3 of C3] {\Huge \bf $\cdots$};

\end{tikzpicture}
\caption{\label{fig:undirconstruction}
 Our exponential aspect ratio lower bound in undirected graphs (Theorem \ref{thm:lbundir}). The graph consists of $n/5$ cycles of length. The cross-cycle edges (gray, dotted) go between \textbf{differently-numbered} nodes in adjacent cycles; the cross-cycle edges have weight $1$. (The cross-cycle edges have weight $1/3^{i-1}$ for the $\alpha$-stretch lower bound.)}


\vspace{9mm}
\begin{tikzpicture}

\def\dist{4.5}
\def\crosscolor{black}

\node [style={circle,fill=blue!20,minimum width=1.2cm}] (A1) at (72*1:1.5) {$v_i^1$};
\node [style={circle,fill=blue!20,minimum width=1.2cm}] (A2) at (72*2:1.5) {$v_i^2$};
\node [style={circle,fill=blue!20,minimum width=1.2cm}] (A3) at (72*3:1.5) {$v_i^3$};
\node [style={circle,fill=blue!20,minimum width=1.2cm}] (A4) at (72*4:1.5) {$v_i^4$};
\node [style={circle,fill=blue!20,minimum width=1.2cm}] (A5) at (72*5:1.5) {$v_1^5$};
\draw (A2) -- (A1);
\draw (A3) -- (A2);
\draw (A4) -- (A3);
\draw [red, line width=0.5em] (A5) -- (A4);
\draw (A1) -- (A5);

\node at (0, -2.5) {edge wts $1/3^i$};

\node [style={circle,fill=blue!20,minimum width=1.2cm}] (B1) [right=5 of A1] {$v_{i+1}^1$};
\node [style={circle,fill=blue!20,minimum width=1.2cm}] (B2) [right=5 of A2] {$v_{i+1}^2$};
\node [style={circle,fill=blue!20,minimum width=1.2cm}] (B3) [right=5 of A3] {$v_{i+1}^3$};
\node [style={circle,fill=blue!20,minimum width=1.2cm}] (B4) [right=5 of A4] {$v_{i+1}^4$};
\node [style={circle,fill=blue!20,minimum width=1.2cm}] (B5) [right=5 of A5] {$v_{i+1}^5$};
\draw (B1) -- (B2);
\draw [line width=0.5em, blue] (B2) -- (B3);
\draw [line width=0.5em, blue] (B3) -- (B4);
\draw (B4) -- (B5);
\draw (B5) -- (B1);
\node at (6, -2.5) {edge wts $1/3^{i+1}$};

\draw[\crosscolor,dotted] (A1) to[bend right=20] (B1);
\draw[\crosscolor,dotted] (A2) to[bend right=20] (B3);
\draw[\crosscolor,dotted] (A3) to[bend right=20] (B5);
\draw[blue, line width=0.25em, dotted] (A4) to[bend right=20] (B2);
\draw[red, line width=0.25em, dotted] (A5) to[bend right=20] (B4);

\end{tikzpicture}
\caption{\label{fig:undirpaths} Just like in our directed lower bound, our analysis uses the fact that the blue shortest path uses $0$ edges from $C_i$ and $2$ edges from $C_{i+1}$, while the alternate non-shortest red path uses $1$ edge from $C_i$ and $0$ edges from $C_{i+1}$. For example, in this diagram, the shortest $v_i^4 \leadsto v_{i+1}^4$ path is $(v_i^4, v_{i+1}^2, v_{i+1}^3, v_{i+1}^4)$ (thick, blue) and not $(v_i^4, v_i^5, v_{i+1}^4)$ (thick, red).} 
\end{center}
\end{figure}

\begin{itemize}
\item We assume that the number of vertices $n$ is divisible by $5$. The graph consists of $n/5$ cycles, each with exactly $5$ vertices. Label the cycles $C_1, ... C_{n/5}$, where cycle $C_i$ consists of vertices $v^1_i, v^2_i\dots, v^5_i$, and edges $(v^1_i, v^2_i)$, $(v^2_i, v^3_i)$, $(v^3_i, v^4_i)$, $(v^4_i, v^5_i)$, and $(v^5_i, v^1_i)$.
\item For all $i$, the cross-cycle edges are as follows: $(v^1_i, v^1_{i+1})$, $(v^2_i, v^3_{i+1})$, $(v^3_i, v^5_{i+1})$, $(v^4_i, v^2_{i+1})$, and $(v^5_i, v^4_{i+1})$. The pattern of these edges is that as the superscript of the first vertex increases by 1, the superscript of the second vertex increases by $2 \pmod 5$. 
\item All edges in $C_i$ have weight $1/3^i$. Meanwhile, all cross-cycle edges have weight 1.
\end{itemize}

\subsection{Analysis}
\label{sec:undir-analysis}

Similar to the directed case, the analysis is broken into two parts. First, we will identify a particular set $\mathcal{P}$ of paths and in \cref{lem:spundir} we will show that they are unique shortest paths in $G$. Then, in \cref{lem:aspectundir} we will show that any reassignment of edge weights of $G$ for which all paths in $\mathcal{P}$ remain shortest has aspect ratio $2^{\Omega(n)}$.

\paragraph{Construction of $\mathcal{P}$.} See the blue path in Figure \ref{fig:undirpaths} for an example of a path in $\mathcal{P}$. All paths in $\mathcal{P}$ go from a vertex in $C_i$ to a vertex in $C_{i+1}$. For every vertex $v^k_i$ add to $\mathcal{P}$ the following 3-edge path from $v^k_i$. Take the cross-cycle edge from $v^k_i$ to the cycle $C_{i+1}$ and let $v^j_{i+1}$ be the other endpoint of the edge. Then take two cycle edges in $C_{i+1}$: first the edge to $v^{j+1\pmod 5}_{i+1}$ and then the edge to $v^{j+2\pmod 5}_{i+1}$.

\begin{lemma}\label{lem:spundir}
Each path in $\mathcal{P}$ is the unique shortest path between its endpoints. 
\end{lemma}
\begin{proof} Fix a path $P\in\mathcal{P}$. $P$ is from $v^k_i$ to $v^{j+2\pmod 5}_{i+1}$. We calculate $w(P)$ as follows: the first term is for the weight-1 cross-cycle edge, while the second term is the sum of the two edge weights in $C_{i+1}$.  \[w(P) = 1+2/3^{i+1}.\]
Let $P'$ be an alternate path with the same endpoints as $P$. We will argue that $w(P')>w(P)$. 

Any path from $C_i$ to $C_{i+1}$ must use a cross-cycle edge, which has weight 1. If $P'$ uses another cross-cycle edge, $P'$ gains an additional weight of 1. If $P'$ uses an edge in $C_i$, $P'$ gains weight $1/3^i$. In either of these cases, we have \[w(P')\geq 1+1/3^i>1+2/3^{i+1} = w(P).\] Thus, the only case left is consider is when $P'$ uses exactly one cross-cycle edge, and does not use any edges from $C_i$. 

The only edge from $v^k_i$ that is not in $C_i$ is the first edge of $P$: $(v^k_i, v^j_{i+1})$, so $P'$ must begin with this edge. From $v^j_{i+1}$, the remainder of $P'$ must be within $C_{i+1}$, since $P'$ does not use another cross-cycle edge. Because $C_{i+1}$ is a cycle, there are exactly two simple paths from $v^j_{i+1}$ to $v^{j+2\pmod 5}_{i+1}$: $P$ takes the one on $2$ edges, while the other path goes in the other direction around the $5$-cycle, so it has 3 edges and is longer than $P$. This completes the proof. 
\end{proof}

\begin{lemma}\label{lem:aspectundir}
Any reassignment of edge weights to $G$ such that all paths in $\mathcal{P}$ remain shortest has aspect ratio $2^{\Omega(n)}$.
\end{lemma}

\begin{proof}
    Let $H=(V,E,w_H)$ be a shortest-paths-preserving graph of $G$. We will show that for all $i\in \{1,\dots,n/5-1\}$, we have $w_H(C_i)> 2\cdot w_H(C_{i+1})$. Since there are $n/5$ cycles, this implies that the aspect ratio of $w_H$ is $2^{\Omega(n)}$.

    First, we claim that for every path $P\in\mathcal{P}$ from $C_i$ to $C_{i+1}$, there is a not-shortest 2-edge path $P'$ with the same endpoints as $P$, such that $P'$ includes one edge in $C_i$ and one cross-cycle edge. ($P$ and $P'$ correspond to the blue and red paths respectively in Figure \ref{fig:undirpaths}.) The reason for this is by the construction of the cross-cycle edges. Recall that the cross-cycle edges are $(v^1_i, v^1_{i+1})$, $(v^2_i, v^3_{i+1})$, $(v^3_i, v^5_{i+1})$, $(v^4_i, v^2_{i+1}), (v^5_i, v^4_{i+1})$, where the pattern is that as the superscript of the first vertex increases by 1, the superscript of the second vertex increases by $2 \pmod 5$. Specifically, $P$ has vertices $v^k_i$, $v^j_{i+1}$, $v^{j+1\pmod 5}_{i+1}$, and $v^{j+2\pmod 5}_{i+1}$, while $P'$ has vertices $v^k_i$, $v^{k+1\pmod 5}_i$, $v^{j+2\pmod 5}_{i+1}$. 
    
    Since $P$ is the unique shortest path between its endpoints in $H$, we know that $w_H(P)<w_H(P')$. Taking the sum of this inequality for all 5 paths in $\mathcal{P}$ that go from $C_i$ to $C_{i+1}$, we get $w_H(C_i) + X > 2w_H(C_{i+i}) + X$, where $X$ is the sum of the weights of the cross-cycle edges from $C_i$ to $C_{i+1}$. The $X$ cancels, so we get $w_H(C_i) > 2w_H(C_{i+i})$, as desired. 
\end{proof}

\subsection{Approximate Version}\label{sec:undirapprox}
We now show that exponential aspect ratio is required even if we allow stretch. But unlike in the directed case, we do not have a lower bound against arbitrary stretch, only against a small constant stretch.
\lbundirap*
\paragraph{High-Level Description:} For the approximate version, the third issue for undirected graphs in the high-level description above becomes more troublesome: namely, that alternate paths can go through multiple levels of cycles and then return back to the original level. For the exact version for undirected graphs, we handled this by setting cross-cycle edges to have very high weight relative to the cycle edges; in particular, we set all of their weights to 1. 
This does not work for the approximate version for the following reason. 

Just as in the directed case (\cref{lem:spdirap}), we want to ensure that each path in $\mathcal{P}$ is the \emph{only} $\alpha$-approximate shortest path between its endpoints in $G$. This is no longer true for undirected graphs if we set the cross-cycle edges to have very large weight. In particular, recall that each path $P \in \mathcal{P}$ consist of one cross-cycle edge and two cycle edges from $C_{i+1}$, while there also exists a corresponding not-shortest path $P'$ between the same endpoints that consist of one cross-cycle edge and one cycle edge from $C_i$. The two paths both use exactly one cross-cycle edge, so if that weight is very large then it will dominate the edge-weights on the cycle, and $P'$ will be an $\alpha$-approximation to $P$



To address this issue, we need to set the weight of the cross-cycle edges carefully; setting these weights either too large or too small does not work. This balancing act is the reason that our lower-bound construction only works for a small constant $\alpha$.

\paragraph{Construction:} The construction for our lower-bound graph $G = (V,E,w)$ as identical to the exact version for undirected graphs, except that each cross-cycle edge $e$ from $C_i$ to $C_{i+1}$ is given weight $w(e) = 1/3^{i-1}$.

\paragraph{Analysis:} We will use the same set $\mathcal{P}$ of paths as the exact undirected version (see Section \ref{sec:undir-analysis}). We will show in \cref{lem:spundirap} that each path in $\mathcal{P}$ is the \emph{only} $13/12$-approximate shortest path between its endpoints in $G$. Since $\alpha\leq 13/12$, this means that for any $\alpha$-approximate shortest-path-preserving graph $H = (V,E,w_H)$ of $G$, each path in $\mathcal{P}$ is the unique shortest path between its endpoints in $H$. Combined with \cref{lem:aspectundir}, this implies that $H$ has aspect ratio $2^{\Omega(n)}$, as desired.
Thus, it remains only to prove \cref{lem:spundirap}:

\begin{lemma}\label{lem:spundirap}
    Each path in $\mathcal{P}$ is the only $13/12$-approximate shortest path between its endpoints in $G$.
\end{lemma}

\begin{proof}
   Fix a path $P\in\mathcal{P}$. $P$ is from $v^k_i$ to  $v^{j+2\pmod 5}_{i+1}$ (where $j$ is such that $(v^k_i, v^j_{i+1})$ is a cross-cycle edge). We calculate $w(P)$ as follows: the first term is for the cross-cycle edge, while the second term is the sum of the two edge weights in $C_{i+1}$: 
   \[w(P)=1/3^{i-1}+2/3^{i+1}.\]
   Let $P'$ be an alternate path with the same endpoints as $P$. We will argue that $w(P')>13/12\cdot w(P)$. 

Any path from $C_i$ to $C_{i+1}$ must use a cross-cycle edge between $C_i$ and $C_{i+1}$, which has weight $1/3^{i-1}$. If $P'$ uses another cross-cycle edge, its next cross-cycle edge is from $C_{i+1}$ to either $C_{i+2}$ to $C_i$, so has weight at least $1/3^i$. This means that we would have: \begin{align*}w(P')&\geq 1/3^{i-1}+1/3^{i}\\&>13/12\cdot(1/3^{i-1}+2/3^{i+1})\\&=13/12\cdot w(P).\end{align*} Thus, suppose $P'$ uses exactly one cross-cycle edge.

Furthermore, each edge in $C_i$ is of weight at least $1/3^i$, so due to exactly the same string of inequalities, if $P'$ used any edge in $C_i$ then we would again have $w(P')>13/12\cdot w(P)$. Thus, suppose $P'$ does not use any edges from $C_i$. 

The only edge from $v^k_i$ that is not in $C_i$ is the first edge of $P$: $(v^k_i, v^j_{i+1})$, so $P'$ must begin with this edge. From $v^j_{i+1}$, recall that $P'$ cannot use another cross-cycle edge so the remainder of $P'$ must be within $C_{i+1}$. Because $C_{i+1}$ is a cycle, there are exactly two simple paths from $v^j_{i+1}$ to $v^{j+2\pmod 5}_{i+1}$: $P$ takes the one on 2 edges and the other one goes the other way around the 5-cycle and has 3 edges. These edges have weight $1/3^{i+1}$, so the total weight of 3 of these edges is $1/3^{i}$. Thus, if $P'$ takes the 3-edge path, exactly the same string of inequalities as above implies $w(P')>13/12\cdot w(P)$.
\end{proof}

\section{Upper Bound for DAGs}\label{sec:dagub}

In this section, we show that in contrast to general directed graphs, any DAG $G$ admits a shortest-path-preserving graph with aspect ratio $O(n)$.

\ub*

Since $G$ is a DAG, we can label the vertices $v_1, ..., v_n$ according to their topolgical order; that is, for every edge $(v_i, v_j) \in E$ we have $j > i$. 

Let $\wg$ be the heaviest edge weight in the input graph $G = (V,E,w)$. We now define a new weight function $w_H$: given any edge $(v_i, v_j)$ in $G$, set
$$w_H(v_i,v_j) = w(v_i, v_j) + \wg \cdot (j - i)$$

We now argue that $H = (V,E,w_H)$ is a shortest-path-preserver of $G$.

\begin{lemma}
\label{lem:dag-price}
Every shortest path in $H$ is a shortest path in $G$ and vice versa.
\end{lemma}

\begin{proof}
This follows directly from the well-known fact that price functions do not change shortest paths; note that $w_H(v_i,v_j)$ are precisely the reduced weights obtained from price function $\phi(v_i) = \wg \cdot i$. See e.g. \cite{Johnson77} for more details on price functions. 

For the sake of completeness, we also prove the Lemma from scratch in Section \ref{app:dag-price} of the appendix.

\end{proof}

\begin{lemma} 
The weight function $w_H$ has aspect ratio at most $n+1$.
\end{lemma}

\begin{proof}
Recall that we defined $\wg$ to be the heaviest edge weight in $G$. For any edge $(v_i,v_j)$ we have $1 \leq j-i < n$. We can thus conclude that the maximum edge weight in $H$ is at most $\wg + n\wg = (n+1)\wg$, while the minimum edge weight in $H$ is at least $\wg$. This implies that $w_H$ has aspect ratio at most $n+1$.
\end{proof}

\section{Lower Bound for DAGs}\label{sec:lbdag}

In this section we will prove an exponential lower bound for two-sided approximation on DAGs:

\lbdagap*

\subsection{Construction}

The lower-bound graph $G=(V,E,w)$ is the $\sqrt{n}\times\sqrt{n}$ grid graph where all horizontal edges are directed to the right and all vertical edges are directed upwards (see \cref{fig:grid}). We will refer to the vertices by their (row, column) coordinates, where the vertex in the top left corner is $(0,0)$. 

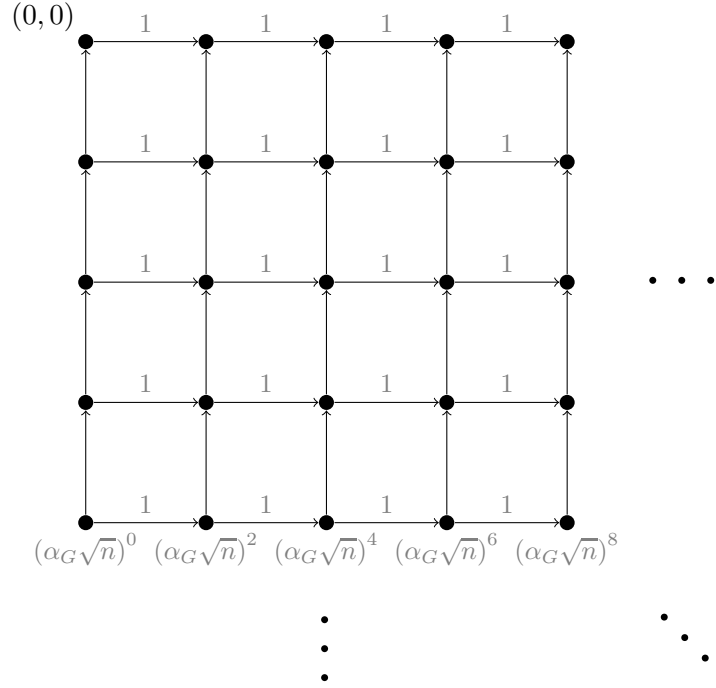
\begin{figure}[h]
\begin{center}
\begin{tikzpicture}[scale=0.8]
    \def\gridsize{4} 

    \foreach \x in {0,1,...,\gridsize}
    {
        \foreach \y in {0,1,...,\gridsize}
        {
            \node[inner sep=1pt,circle,fill,minimum width=0.2cm] (n-\x-\y) at (2*\x,2*\y) {};
        }
    }
    \node [above left] at (0, 2*\gridsize) {$(0, 0)$};

    \foreach \x in {0,1,...,\gridsize}
    {
        \foreach \y in {0,1,...,\numexpr\gridsize-1\relax}
        {
            \draw[->] (n-\x-\y) -- (n-\x-\the\numexpr\y+1\relax);
        }
    }

    \foreach \x in {0,1,...,\numexpr\gridsize-1\relax}
    {
        \foreach \y in {0,1,...,\gridsize}
        {
            \draw[->] (n-\x-\y) to node [midway, above, gray] {\small $1$} (n-\the\numexpr\x+1\relax-\y);
        }
    }

    \foreach \x in {0,1,2,...,\gridsize}
    {
        \pgfmathtruncatemacro\twox{2*\x}
        \node [below, gray] at (2*\x, 0) {\small $\left(\alpha_G \sqrt{n}\right)^{\twox}$};
    }


    \node [rotate=90] at (\gridsize, -2) {\Huge $\cdots$};
    \node at (2*\gridsize+2, \gridsize) {\Huge $\cdots$};
    \node [rotate=-45] at (2*\gridsize+2, -2) {\Huge $\cdots$};

\end{tikzpicture}
\caption{\label{fig:grid} The graph $G$. The weights under each column indicate that every edge in that column has that weight.}
\end{center}
\end{figure}

The edge weights are defined as follows. All horizontal edges have weight $1$. The weight of each vertical edge depends on its column. All vertical edges in column $j$ have weight $(\alpha_G\cdot \sqrt{n})^{2j}$. 

This completes the construction of $G$. Note that $G$ is a DAG.

\subsection{Analysis} The analysis is broken into two parts. First, in \cref{lem:hor,lem:vert} we will show that each path from a particular collection of paths is the \emph{only} $\alpha_G$-approximate shortest path between its endpoints. This means that in any $(\alpha_H\to\alpha_G)$-stretch shortest-paths preserving graph $H=(V,E,w_H)$, each of these paths must be the \emph{only} $\alpha_H$-approximate shortest path between its endpoints. Then, in \cref{lem:rowcol} we will prove by induction on the size of the grid, that the sum of edge weights in the last row and the last column of $H$ must be exponentially large.

\begin{lemma} \label{lem:hor}
Let $s=(i,j)$ be a vertex in $G$, and let $t=(i-1,k)$ where $k>j$ (that is, $t$ is exactly one row above and at least one column to the right of $s$). Then, the shortest $st$-path $P$ uses one vertical edge followed by a series of horizontal edges. Furthermore, $P$ is the only $\alpha_G$-approximate $st$-shortest path.
\end{lemma}
\begin{proof}
    Since edges are only directed up and to the right, and $t$ is exactly one row above $s$, the only $st$-paths use a (possibly empty) series of horizontal edges, followed by a single vertical edge, followed by a (possibly empty) series of horizontal edges. Since all horizontal edges have weight $1$, the contribution of these edges to $w(P)$ is $k-j$. Let $\ell$ be the column of the single vertical edge $e$ in $P$ (so $j\leq \ell\leq k$). Then, $w(e)=(\alpha_G\cdot \sqrt{n})^{2\ell}$. This quantity is minimized when $\ell$ is minimized, so $\ell=j$. This means that the shortest $st$-path $P$ has its vertical edge in column $j$ (the column containing $s$). This proves the first part of the claim.

    To prove that $P$ is the only $\alpha_G$-approximate $st$-shortest path, we consider any other $st$-path $P'$. Let $\ell'>j$ be the column of the single vertical edge in $P'$. Then, $w(P')=k-j+(\alpha_G\cdot \sqrt{n})^{2\ell'}\geq (\alpha_G\cdot \sqrt{n})^{2(j+1)}$. 
    Then,\begin{align*}\frac{w(P')}{w(P)}&\geq \frac{(\alpha_G\cdot \sqrt{n})^{2(j+1)}}{k-j+(\alpha_G\cdot \sqrt{n})^{2j}}\\
 &>\frac{(\alpha_G\cdot \sqrt{n})^{2(j+1)}}{(\alpha_G\cdot \sqrt{n})^{2j+1}}\\
    &\geq \alpha_G.
    \end{align*}
\end{proof}

We also make a symmetric claim:

\begin{lemma} \label{lem:vert}
Let $s=(i,j)$ be a vertex in $G$, and let $t=(k,j+1)$ where $k<i$ (that is, $t$ is at least one row above and exactly one column to the right of $s$). Then, the shortest $st$-path $P$ uses a series of vertical edges followed by exactly one horizontal edge. Furthermore, $P$ is the only $\alpha_G$-approximate $st$-shortest path.
\end{lemma}
\begin{proof}
    Since edges are only directed up and to the right, and $t$ is exactly one column to the right of $s$, the only $st$-paths use a (possibly empty) series of vertical edges, followed by a single horizontal edge, followed by a (possibly empty) series of vertical edges. The contribution of the single horizontal edge to $w(P)$ is 1. The contribution of the vertical edges is as follows. Let $b$ be the number of vertical edges on $P$ before the horizontal edge. Then the sum of the weights of the vertical edges is $b(\alpha_G\cdot \sqrt{n})^{2j} + (i-k-b)(\alpha_G\cdot \sqrt{n})^{2(j+1)}$. This quantity is minimized when $b$ is maximized, so $b=i-k$. This proves the first part of the claim.

    To prove that $P$ is the only $\alpha_G$-approximate $st$-shortest path, we consider any other $st$-path $P'$. Let $b'<i-k$ be the number of vertical edges on $P'$ before the horizontal edge. Then, \begin{align*}w(P')&=1+b' (\alpha_G\cdot \sqrt{n})^{2j} + (i-k-b') (\alpha_G\cdot \sqrt{n})^{2(j+1)} \\&\geq (i-k-1) (\alpha_G\cdot \sqrt{n})^{2j} + (\alpha_G\cdot \sqrt{n})^{2(j+1)}\\&\geq(\alpha_G\cdot \sqrt{n})^{2(j+1)}.\end{align*}

We know that $w(P)=1+ (k-i) (\alpha_G\cdot \sqrt{n})^{2j}$. Thus,\begin{align*}\frac{w(P')}{w(P)}&\geq \frac{(\alpha_G\cdot \sqrt{n})^{2(j+1)}}{1+ (i-k) (\alpha_G\cdot \sqrt{n})^{2j}}\\
  &\geq \frac{(\alpha_G\cdot \sqrt{n})^{2(j+1)}}{1+ \sqrt{n} (\alpha_G\cdot \sqrt{n})^{2j}}\\
    &> \frac{(\alpha_G\cdot \sqrt{n})^{2(j+1)}}{(\alpha_G\cdot \sqrt{n})^{2j+1}} \text{\hspace{1em} for sufficiently large $n$}\\
    &\geq \alpha_G.
    \end{align*}
\end{proof}

Now we will prove that in $H$, the sum of last column and the last row of the grid must have exponentially large weight. We can assume without loss of generality that the minimum edge weight in $H$ is 1. 

\begin{lemma}\label{lem:rowcol}
    Suppose that the minimum edge weight in $H$ is $1$. Then, the sum of edge weights (in $H$) of the last row and last column is at least $(\alpha_H)^{\sqrt{n}-1}$.
\end{lemma}

\begin{proof}
    The proof is by induction on the dimension of the grid. We will show that for any $L$, when $H$ is an $L\times L$ grid, the sum of weights in the last row and the last column is at least $(\alpha_H)^{L-1}$.
    
    \paragraph{Base Case.} In the base case $H$ is a $2\times 2$ grid. Let $P'$ be the path $(1,0)\rightarrow (1,1)\rightarrow (0,1)$; that is $w_H(P')$ is the sum of weights in the last row and column. Our goal is to show that $w_H(P')\geq \alpha_H$ . Let $P$ be the shortest path from $(1,0)$ to $(0,1)$. By \cref{lem:hor}, $P$ uses one vertical edge followed by one horizontal edge; that is, $P=(1,0)\rightarrow (0,0)\rightarrow (0,1)$. By \cref{lem:hor}, we also know that $P'$, which has the same endpoints as $P$, is not an $\alpha_G$-approximate shortest path in $G$, and thus cannot be an $\alpha_H$ approximate shortest path in $H$. Since the minimum edge weight in $H$ is 1, $w_H(P)\geq 2$, so $w_H(P')\geq 2\cdot \alpha_H$, as desired.
    
    \paragraph{Inductive Hypothesis.} Suppose that when $H$ is an $L\times L$ grid, the sum of weights in the last row and the last column is at least $(\alpha_H)^{L-1}$.

    \paragraph{Inductive Step.} We will show that when $H$ is an $(L+1)\times (L+1)$ grid, the sum of edge weights in the last row and the last column is at least $(\alpha_H)^{L}$.

Let $H_L\subseteq H$ be the $L\times L$ grid that excludes the last row and last column of $H$. We will apply the inductive hypothesis on $H_L$.

We define the paths $X$, $X'$, $Y$, and $Y'$ as shown in \cref{fig:X}. $X$ is  the last row of $H_L$. $X'$ is $X$ shifted down by one row. $Y$ is the last column of $H_L$. $Y'$ is $Y$ shifted to the right by one column. That is, 
\begin{itemize}[itemsep=0pt]
\item $X$ is the horizontal path $(L-1,0)\rightarrow(L-1,L-1)$, 
\item $X'$ is the horizontal path $(L,0)\rightarrow (L,L-1)$, 
\item $Y$ is the vertical path $(L-1,L-1)\rightarrow(0,L-1)$, and 
\item $Y'$ is the vertical path $(L-1,L)\rightarrow(0,L)$.
\end{itemize}

\begin{figure}[h!]
\begin{center}
\begin{tikzpicture}[scale=0.8]
    \def\gridsize{8} 

    \foreach \x in {0,1,...,\gridsize}
    {
        \foreach \y in {0,1}
        {
            \node[inner sep=1pt,circle,fill,minimum width=0.2cm] (n-\x-\y) at (\x,\y) {};
        }
    }
    \foreach \x in {\the\numexpr\gridsize-1\relax,\gridsize}
    {
        \foreach \y in {0,1,...,\gridsize}
        {
            \node[inner sep=1pt,circle,fill,minimum width=0.2cm] (n-\x-\y) at (\x,\y) {};
        }
    }
    \node[inner sep=1pt,circle,fill,minimum width=0.2cm] (n-0-\gridsize) at (0,\gridsize) {};
    
    \node [above left] at (0, \gridsize) {$(0, 0)$};
    \node [below left] at (0, 0) {$(L, 0)$};
    \node [above right] at (\gridsize, \gridsize) {$(0, L)$};
    \node [below right] at (\gridsize, 0) {$(L, L)$};

    \foreach \y in {1,2,...,\the\numexpr\gridsize-1\relax}
    {
        \draw[->,blue] (n-\the\numexpr\gridsize-1\relax-\y) -- (n-\the\numexpr\gridsize-1\relax-\the\numexpr\y+1\relax);
    }
    \foreach \y in {1,2,...,\the\numexpr\gridsize-1\relax}
    {
        \draw[->,red] (n-\gridsize-\y) -- (n-\gridsize-\the\numexpr\y+1\relax);
    }

    \foreach \x in {0,1,...,\numexpr\gridsize-2\relax}
    {
        \draw[->, red] (n-\x-0) to (n-\the\numexpr\x+1\relax-0);
    }
    \foreach \x in {0,1,...,\numexpr\gridsize-2\relax}
    {
        \draw[->, blue] (n-\x-1) to (n-\the\numexpr\x+1\relax-1);
    }

    \draw [->, blue] (n-0-0) -- (n-0-1);
    \draw [->, blue] (n-\the\numexpr\gridsize-1\relax-\gridsize) -- (n-\gridsize-\gridsize);
    \draw [->, red] (n-\the\numexpr\gridsize-1\relax-0) -- (n-\the\numexpr\gridsize-1\relax-1);
    \draw [->, red] (n-\the\numexpr\gridsize-1\relax-1) -- (n-\gridsize-1);
    
    \draw[decorate,decoration={brace,amplitude=10pt,mirror,raise=2ex}, red]
    (0,0) -- node[below=8ex] {$X'$} (\the\numexpr\gridsize-1,0);

    \draw[decorate,decoration={brace,amplitude=10pt,mirror,raise=2ex}, red]
    (\the\numexpr\gridsize,1) -- node[midway, right=8ex] {$Y'$} (\the\numexpr\gridsize,\gridsize);

    \draw[decorate,decoration={brace,amplitude=10pt,raise=1ex}, blue]
    (0, 1) -- node[midway, above=4ex] {$X$} (\the\numexpr\gridsize-1\relax,1);

    \draw[decorate,decoration={brace,amplitude=10pt,raise=1ex}, blue]
    (\the\numexpr\gridsize-1\relax,1) -- node[midway, left=4ex] {$Y$} (\the\numexpr\gridsize-1\relax,\gridsize);

\end{tikzpicture}
\caption{The paths $X$, $X'$, $Y$, and $Y'$. The shortest path $P_X$ is the blue path from $(L,0)$ to $(L-1,L-1)$ (a vertical edge followed by the path $X$), while the not shortest path $P_{X'}$ is the red path with the same endpoints (the path $X'$ followed by a vertical edge). The shortest path $P_Y$ is the blue path from $(L-1,L-1)$ to $(0,L)$ (the path $Y$ followed by a horizontal edge), while the not-shortest path $P_{Y'}$ is the red path with the same endpoints (a horizontal edge followed by the path $Y'$). The remaining grid edges are omitted from the figure for simplicity.}
\label{fig:X}
\end{center}
\end{figure}

By the inductive hypothesis we know that
$w_H(X)+w_H(Y) > (\alpha_H)^{L-1}$.

Let $P_X$ be the shortest path in $H$ from $(L,0)$ to $(L-1,L-1)$ (see Figure \ref{fig:X}). By \cref{lem:hor}, $P_X$ 
takes one vertical edge followed by all of the edges in $X$. Consider an alternate path $P_{X'}$ with the same endpoints, which takes all of the edges in $X'$ followed by the vertical edge $(L,L-1)\rightarrow (L-1,L-1)$. By \cref{lem:hor}, $P_{X'}$ is not an $\alpha_G$-approximate shortest path in $G$, and thus cannot be an $\alpha_H$ approximate shortest path in $H$. Thus, \begin{equation}\label{eq:hor}w_H(X')+w_H\big((L,L-1),(L-1,L-1)\big) > \alpha_H\cdot w_H(X).\end{equation}

Now, we will derive a symmetric inequality for $Y$ instead of $X$. Let $P_Y$ be the shortest path in $H$ from $(L-1,L-1)$ to $(0,L)$. By \cref{lem:vert}, $P_Y$ 
takes all of the edges in $Y$ followed by one horizontal edge. Consider an alternate path $P_{Y'}$ with the same endpoints, which takes the horizontal edge $(L-1,L-1)\rightarrow (L-1,L)$ followed by all of the edges in $Y'$. By \cref{lem:vert}, $P_{Y'}$ is not an $\alpha_G$-approximate shortest path in $G$, and thus cannot be an $\alpha_H$ approximate shortest path in $H$. Thus, \begin{equation}\label{eq:vert}w_H(Y')+w_H\big((L-1,L-1),(L-1,L)\big) > \alpha_H\cdot w_H(Y).\end{equation}

Taking the sum of \cref{eq:hor,eq:vert}, we have 
\begin{equation}\label{eq:comb}w_H(X')+w_H(Y')+w_H\big((L,L-1),(L-1,L-1)\big)+w_H\big((L-1,L-1),(L-1,L)\big) > \alpha_H\cdot (w_H(X)+ w_H(Y))\end{equation}

By \cref{lem:hor}, we know that $(L,L-1)\rightarrow(L-1,L-1)\rightarrow(L-1,L)$ is a shortest path, while $(L,L-1)\rightarrow(L,L)\rightarrow(L-1,L)$ is not. Thus, the left hand side of \cref{eq:comb} is strictly less than $w_H(X')+w_H(Y')+w_H\big((L,L-1),(L,L)\big)+w_H\big((L,L),(L-1,L)\big)$, which is precisely the sum of the last row and last column of $H$. Thus, our goal is to show that the right hand side of \cref{eq:comb} is at least $(\alpha_H)^{L}$.

By the inductive hypothesis, $w_H(X)+w_H(Y) > (\alpha_H)^{L-1}$, so we have \begin{align*}\alpha_H\cdot (w_H(X)+ w_H(Y))&>\alpha_H \cdot (\alpha_H)^{L-1}\\&=(\alpha_H)^{L}.\end{align*}

This completes the proof.
\end{proof}

We have just proven in \cref{lem:rowcol} that in $H$ the sum of edge weights in the last row and last column is at least $(\alpha_H)^{\sqrt{n}-1}$. This means that at least one of these $2\sqrt{n}$ edges has weight at least $(\alpha_H)^{\sqrt{n}-1} / (2\sqrt{n}) = (\alpha_H)^{\Omega(\sqrt{n})}$. This completes the proof of \cref{thm:lbdagap}.





\bibliographystyle{alpha}
 \bibliography{references}

\clearpage

\begin{appendices}

\section{Proof of Lemma \ref{lem:dag-price}}
\label{app:dag-price}

\begin{proof}

Consider any paths $P,P'$ between the same set of endpoints. We will show that $w(P) - w(P') = w_H(P) - w_H(P')$, which clearly implies the lemma.

Let $P = v_{i_1}, v_{i_2}, ..., v_{i_k}$. We have 
$$w(P) = \sum_{j=1}^{k-1} w(v_{i_j},v_{i_{j+1}})$$ 
We also have 
$$w_H(P) = \sum_{j=1}^{k-1} [w(v_{i_j},v_{i_{j+1}}) + \wg(i_{j+1} - i_j)] = \sum_{j=1}^{k-1} w(v_{i_j},v_{i_{j+1}}) + \wg(i_k - i_1),$$
where the second inequality holds because the sum $\sum_{j=1}^{k-1} \wg(i_{j+1} - i_j)$ telescopes. 

Combining the inequalities above yields $w_H(P) = w(P) + \wg(i_k - i_1)$. By an identical argument we $w_H(P') = w(P') + \wg(i_k - i_1)$; here we use the fact that $P$ and $P'$ have the same endpoints, namely $v_{i_1}$ and $v_{i_k}$. Subtracting these two equalities, we conclude that
\begin{align*}
w(P) - w_H(P) = w(P') - w_H(P') \rightarrow w(P) - w(P') = w_H(P) - w_H(P'). \tag*{\qedhere}
\end{align*}
\end{proof}

\end{appendices}

 \end{document}